\documentclass[openacc]{rstransa}
\usepackage[english]{babel}
\usepackage{subcaption}
\usepackage{hyperref}

\def\Tr{\operatorname{Tr}}
\def\trnsfrm#1{\mathcal #1}
\def\tA{\trnsfrm A}\def\tB{\trnsfrm B}\def\tC{\trnsfrm C} 
 
\def\tI{\trnsfrm I}\def\tT{\trnsfrm T}

 \def\tT{\trnsfrm T} 
\def\rA{{\rm A}}\def\rB{{\rm B}}\def\rC{{\rm C}} \def\rE{{\rm E}} 
\def\rI{{\rm I}}

\def\rX{{\rm X}}\def\rY{{\rm Y}} 
\def\sH{{\mathcal{H}}}
\def\St{\rm{St}}\def\Eff{\mathrm{Eff}}\def\Trn{\mathrm{Transf}}\def\T{\mathrm{T}}\def\Bnd{\mathrm{Bnd}}
\def\Cmplx{\mathbb{C}}

\begin{document}
\newtheorem{theorem}{Theorem}[section]
\newtheorem{corollary}[theorem]{Corollary}
\newtheorem{definition}[theorem]{Definition}
\newcommand*{\doi}[1]{\href{http://dx.doi.org/#1}{doi: #1}}

\title{Causality re-established}
\author{Giacomo Mauro D'Ariano$^{1,2}$}

\address{$^{1}$Dipartimento di Fisica, Universit\`a degli Studi di Pavia, via Bassi 6, 27100 Pavia,\\
$^{2}$ INFN, Gruppo IV, Sezione di Pavia}

\subject{Physics and Philosophy}


\corres{Giacomo Mauro D'Ariano\\
\email{dariano@unipv.it}}
\begin{abstract}
Causality never gained the status of a ``law'' or ``principle'' in physics. Some recent literature even popularised the false idea that causality is a notion that should be banned from theory. Such misconception  relies on an alleged universality of reversibility of  laws of physics, based either on determinism of classical theory, or on the multiverse interpretation of quantum theory, in both cases motivated by mere interpretational requirements for realism of the theory. 

Here I will show that a properly defined unambi-guous notion of causality is a theorem of quantum theory, which is also a falsifiable proposition of the theory. Such causality notion appeared in the literature within the framework of operational probabilistic theories. It is a genuinely theoretical notion, corresponding to establish a definite partial order among events, in the same way as we do by using the future causal cone on Minkowski space. 

The causality notion is logically completely indepen-dent on the misidentified concept of  ``determinism'', and, being a consequence of quantum theory, is ubiquitous in physics. In addition, since classical theory can be regarded as a restriction of quantum theory, causality holds also in the classical case, although the determinism of the theory trivialises it.

I then conclude  arguing that causality naturally establishes an arrow of time. This  implies that the  scenario of the ``Block Universe'' and the connected ``Past Hypothesis'' are incompatible with causality, and thus with quantum theory: they both  are doomed to remain mere interpretations and, as such, not falsifiable, similarly to the  hypothesis of ``super-determinism''.
\end{abstract}
\begin{fmtext}
The present opinion paper should be regarded as a manifesto for more in-depth investigations for such a relevant notion in sciences.
\end{fmtext}
\maketitle
\section{Causality, the Cinderella of Physics}

Causality has been always an undetermined and controversial notion in physics, perhaps due to its involvement in a wide spectrum of heterogeneous disciplines, including all natural sciences, pure philosophy, economics, and law.\footnote{The literature on causality is very extensive, due to the number of disciplines in which it is involved. Perhaps the most natural connection between concepts of causality in different branches of knowledge is the one at the borderline between physics and philosophy, starting from the early work of Aristotle, up to the Ren\'ee Descartes, who broke the ground for the modern view of David Hume and Immanuel Kant. A modern philosopher who focused on causality was Wesley Salmon, who made several attempts to discriminate ``cause'' from ``effect'' in physical causation, without relying on the arrow of time (see e.g. his book \cite{salmon1998causality} and his articles in Ref. \cite{sosa1993causation}).  Later Phil Dowe devoted his book \cite{dowe2007physical} to a critical analysis of all such attempts, and proposed alteratives. In all these studies an hidden issue is the {\em objectivity} of the assessment of which is the cause and which the effect, whereas, as shown in the present paper, causality is a genuinely {\em theoretical} notion, and depends on our theory of connections among events. The causality notion is also affected by the same objectivity bias about the same of probability notion when causality is considered in a probabilistic context. In this sense it is also the core of Hume's problem of induction. Additional  interesting recent books on the subject of causality are Refs. \cite{bunge1979causality,leung2002causation,cartwright2007hunting,harriman2010the}.}
The controversial status of causality in modern physics is witnessed by some mainstream divulgations, e.g. the recent book \cite{bigpict} of Sean M. Carrol, where on pag. 62 he discredits the notion of causality  in physics with the following  lines
\\\par {\em \ldots we highlighted how Laplace's conservation of information undermines the central role that Aristotle placed on causality. Concepts like  ``cause''  appear nowhere in Newton's equations, nor in our more modern formulations of the laws of nature. But we can't deny that the idea of one event being caused by another is very natural, and seemingly a good fit to how we experience the world. }
 \\\par And then, to add credit to his assertion, in the following page Carrol quotes the famous passage by Bertrand Russel \cite{Russel:1912}:
 \\\par {\em The law of causality, I believe, like much that passes muster among philosophers, is a relic of a bygone age, surviving, like the monarchy, only because it is erroneously supposed to do no harm.}
 \\\par  
Carrol purports to convince the laymen that causality, although used at any instant of our everyday life, nevertheless it is not a law of nature, but only a practical idea, and it  is obsolete as a theoretical tool. On the same lines is the article by J. D. Norton paper, whose title ``Causation as folk science'' technically means that theorising in terms of {\em cause and effects} is a very primordial way of reasoning in science, however, it does not correspond to a precise law or principle in physics at the very fundamental level (see also Ref. \cite{anti-fund}). This is evident in the following lines by Norton:
\\\par  {\em \ldots the concepts of cause and effect are not the fundamental concepts of our science and that science is not governed by a law or principle of causality \cite{folksci}.} 
\\\par
The misidentification between causal reasoning and the existence of a causality principle seems to be the reason why causality has remained in the realm of philosophy, and never achieved the status of  ``principle''. Nevertheless, causality creeps in the physical theories in the form of {\em ad hoc} assumptions based on empirical evidences, as when advanced potentials are discarded in electrodynamics, or when we motivate the Kramers-Kronig relations. In other cases, causality is embodied in the interpretation of the theory, as is the case of Special Relativity. Or else, causality is so naturally embedded in our theoretical understanding, that its notion remains unconsciously implicit even in recent axiomatised quantum foundations, as it happened in the seminal paper  by  Lucien Hardy \cite{hardy2001}, which inspired the informational derivation of quantum theory \cite{purification,QUIT-Arxiv,CUPDCP} of which Causality is now an axiom.\footnote{Later in Ref. \cite{Hardy:2011un} Hardy recognised the role of the causality principle of Refs. \cite{QUIT-Arxiv,CUPDCP}.}

It should be noted that the denial of existence of a principle of causality is generally motivated within a classical theoretical view, as it is also evident in the above Carrol's citation. As explained in his book \cite{bigpict}, the reason beneath such denial within classical physics is the time reversal symmetry of  Newton's and Maxwell's laws. A possible source of misunderstanding here may be the implicit  assumption that the absence of a law would constitute a law by itself. However, both Newton's and Maxwell's theories are compatible with an additional law that picks a preferred direction of time, since one may equally state the same laws in integral form by means of the retarded potentials only, and using them as the law formulation, from which Newton's and Maxwell's original laws can be derived by differentiation, however, with the additional time-arrow restriction. 

A second reasoning underlying the denial of a causality principle, again of a classical-theoretical origin, is the same fact that classical mechanics is deterministic, and the one-to-one correspondence between cause and effect without assuming the arrow of time (the effect occurs after the cause) puts the two notions in one-to-one correspondence, making them interchangeable. Such reasoning is also the origin of the traditional misidentification of the notion of causality with that of determinism--whereas, instead, the two are logically utterly independent.\footnote{The misidentification of the two notions of ``causality'' and ``determinism'' is  hidden ubiquitously in the literature. For example, it is the reason why the de Broglie-Bohm pilot-wave theory is also called {\em causal interpretation of quantum mechanics}. The logical independence between the two notions of causality and determinism is proved by the existence of causal theories that are not deterministic (e.g. quantum theory), and viceversa of deterministic theories that are  not causal (e.g. the OPT presented in\cite{d2013determinism}). (In a probabilistic context ``determinism'' is identified with the tautological property of a theory of having all probabilities of physical events being equal to either zero or one, which is clearly a definition with no causal connotation. Originally, the notion of determinism arose within the clockwork-universe vision of classical mechanics, assessing that the state of a system at an initial time completely determines the state at any later time.)

While in classical theory the notions of causality and determinism degenerate in a conceptual overlapping, they completely disentangle in quantum theory (and, more generally, in an OPT \cite{purification,QUIT-Arxiv}). This is due to the fact that classical mechanics identifies the {\em state} (point in the phase-space) with the {\em measurement-outcome}, while the two concepts are radically different in quantum theory, and more generally in OPT's, allowing us to define determinism outside the framework of classical mechanics which is already deterministic.} And, indeed, the two notions are so deeply entangled in some literature, that they are often merged into the commonplace of {\em causal determinism}. An example of such identification is the quotation by the founding father of quantum theory, Max Planck:
 \\\par\emph{An event is causally determined if it can be predicted with certainty \cite{planck1941kausalbegriff}.} 
\\\par The confusion between causality and determinism is notoriously the main source of the common misinterpretation of perfect EPR quantum correlations (e.g. from a singlet), that originated the Einstein's motto ``spooky action at a distance'', whose literal interpretation corresponds to the subliminal interpretation of the correlations being causal.
\\\par The advent of Quantum Mechanics has led us to contemplate causal relations as generally probabilistic, as it is the case in all natural and human sciences. Although a probabilistic context for physical causation had already been considered by several authors \cite{salmon1998cauality,dowe2007physical,Pearl:2012},
the precise mathematical formulation of causality has been given only recently in the context of {\em operational probabilistic theories} (OPT) \cite{purification} (see Sect. \ref{s:causal}). Such formulation covers all possible instances of the notion, and corresponds to  set up a partial ordering among events, thus requiring a network description, which is embedded in the operational framework. The OPT formulation is perfectly compliant with the historical philosophical concept of causality since David Hume, along with its usage in common reasoning, inference, and modelling in sciences \cite{Pearl:2012}. From its OPT formulation we can also appreciate that, being an ingredient of a probabilistic logical inference approach,   causality is a genuinely theoretical notion, and this simple fact closes all debates about ``objective'' identification of causes versus effects without assuming the arrow of time, as in the case in the Phil Dowe analysis \cite{dowe2007physical} of Wesley Salmon attempts\cite{salmon1998causality,sosa1993causation}. Even the same issue of Hume's induction corresponds to such misinterpretation of causality as an objective notion, whereas, instead, causality  is a purely theoretical ingredient, and, being probabilistic, it depends on prior assumptions.\footnote{In this sense the criticism of Bertrand Russell is pertinent to such objective formulations, along with the purely deterministic notion, and all other criticisms are completely overcome by the OPT formulation.}

In the next section we will see how causality is a theorem of quantum theory, and, as such, is a falsifiable proposition of the theory. We will also see how the present notion of causality is equivalent to the Einsteinian notion, and distills all the guises in which causality appears in physics.

\section{Causality is a theorem of Quantum Theory}\label{s:causal}
Quantum theory is an operational probabilistic theory (OPT) \cite{CUPDCP}. An OPT is a theory that can predict joint probability distributions of multiple events depending on a graph of connections between them. The theory associates mathematical structures to both  ``nodes'' and ``links'' of the graph, which are {\em transformations} and the {\em systems} of the theory, respectively.  

We will now report a presentation of quantum theory as OPT based on von the Neumann axiomatisation, and then state the causality principle, proving it as a theorem of the theory. Then, we will see how the principle corresponds to set up a partial ordering among events, and it is equivalent to the customary causality notion of special relativity. Finally, we will establish the scientific value of the principle on the basis of its falsifiability. 

\subsection{Quantum Theory as operational probabilistic theory}
Quantum theory is an operational probabilistic theory (OPT), namely a theory that predicts joint probability distributions of events depending on  a graph of connections among them. 
The graph is what is generally referred to as ``quantum circuit'' in the quantum information literature. It is a DAG (acronym for ``directed acyclic graph''), here conventionally directed from the left to the right. An example of DAG represented as a customary quantum circuit is given in Fig. \ref{f:DAG}. The nodes of the graph are the {\em events} (in figure denoted as $\Psi_i,\tA_j,\ldots$). The links of the graph are the {\em systems} of the theory (denoted as $\rA, \rB,\ldots$). The collection of all possible alternative events ${\tA_j}$ with overall unit probability is called {\em test}, and denoted as $\{\tA_j\}_{j\in\rX}$ with $\rX$ (finite countable) probability space.\footnote{For those not familiar with quantum information, I would emphasise that the graph description can equally represent an experiment made in the lab, or a ``natural process'', e.g. corresponding to a Feynman diagram.} By definition the test is complete if the coarse-grained event $\tA_\rX:=\sum_{j\in\rX}\tA_j$ is deterministic, namely it occurs with unit marginal probability.\footnote{A deterministic transformation is called {\em quantum channel} in the quantum information literature.}

According to the graph orientation, the wires on the left side of an event box are the {\em input systems}, those on the right side the {\em output systems}. There are special events with trivial input system, called {\em preparations} (represented by boxes with the left rounded side), and events with trivial output system, called {\em observations} (represented by boxes with the right rounded side). When needed we will denote the trivial system by $\rI$. The complete circuit is therefore a closed DAG, as in Fig. \ref{f:DAG}. 
\begin{figure}[h]
\begin{center}
\includegraphics[width=.5 \textwidth]{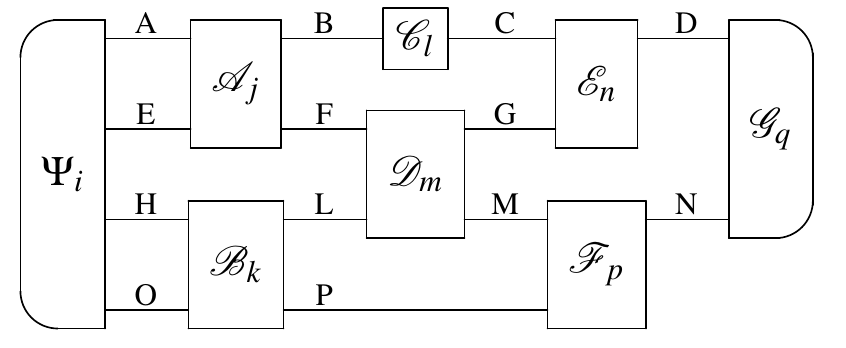}\\
\caption{A closed quantum circuit representing theoretical connections between events. The circuit nodes (depicted as boxes) represent events. The circuit links represent systems of the theory. The graph is a closed directed acyclic graph (DAG), undirected from left to right (input systems for each event are on the left of the box, output systems on the right). The joint probability of all events depends parametrically on the graph. The same graph can be used to represent the complete joint test, upon substituting each event, say ${\tA_j}$ with the test 
$\{\tA_j\}_{j\in\rX}$ with probability space $\rX$.}\label{f:DAG}
\end{center}
\end{figure}
             
We can compose events in parallel and in sequence to built up composite events, as we normally do with quantum circuits, and contextually we  compose systems shared by a sequence of two events, e.g. systems $\rA$ and $\rE$ in Fig. \ref{f:DAG} into the composite system $\rA\rE$ in Fig.\ref{f:DAG}.\footnote{For a detailed axiomatic formulation of composition rules of events and systems see Ch. 4 of  book \cite{CUPDCP}.}
\par We now focus only on the probabilistic nature of the event,  abstracting from its specific instance, and take the equivalence class of events that occurs with the same joint probability in all possible circuits. We call such equivalence class {\em transformation}, and for preparation events we call it {\em state}, and for observation events we call it {\em effect}. By definition a state followed at the output by a transformation is a new state, and similarly an effect preceded at its input by a transformation is a new effect. In this way we can easily understand that every closed circuit can be regarded as the composition of a state with an effect (or a composition of a preparation with an observation): see Fig.  \ref{f:EQU}.
\begin{figure}[h]
     \centering
    \begin{subfigure}[t]{0.47\textwidth}
        \raisebox{0pt}{\includegraphics[width=\textwidth]{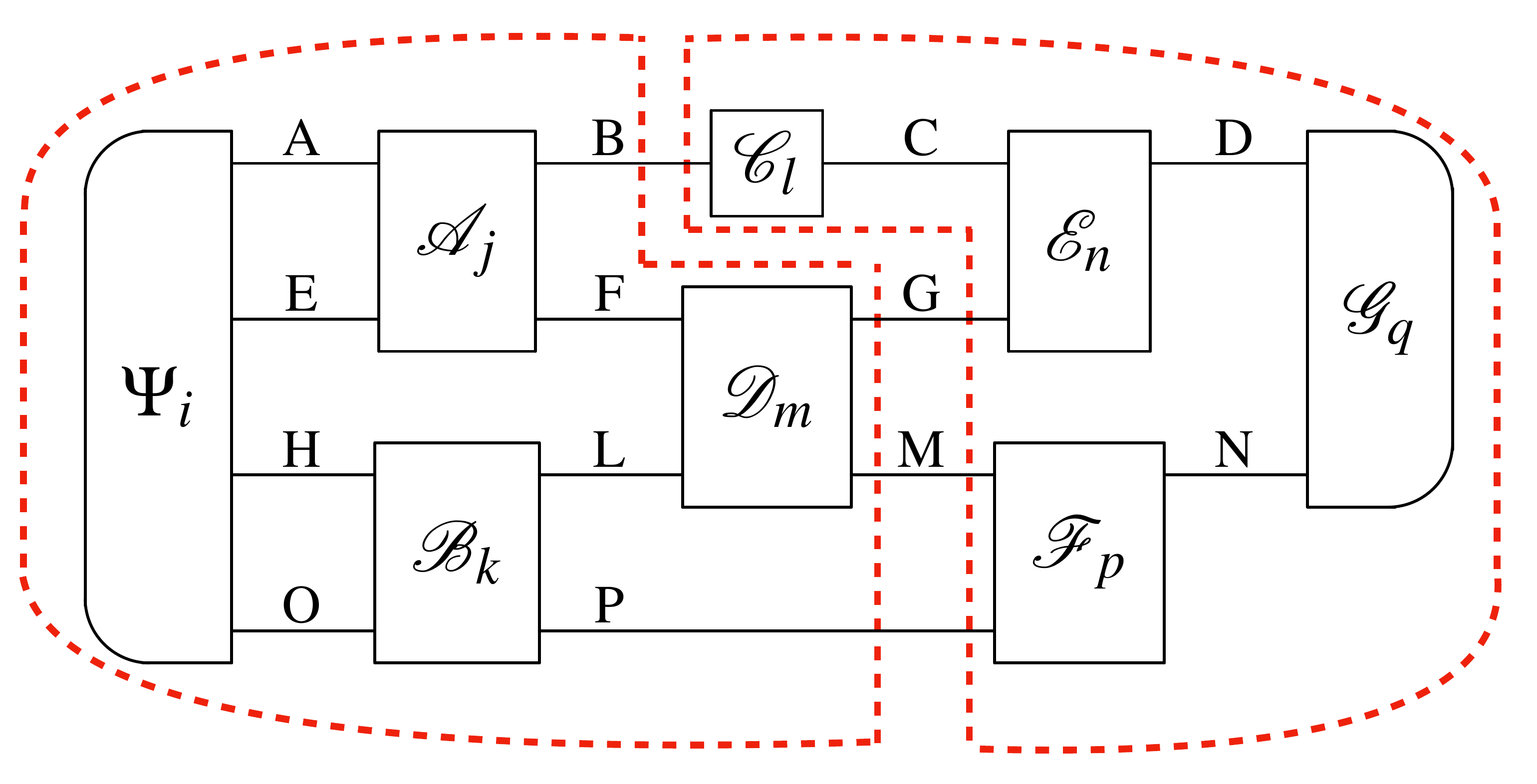}}
    \end{subfigure}
    \hfill
    \begin{subfigure}[t]{0.10\textwidth}  \raisebox{40pt}{\huge{$\equiv$}}
    \end{subfigure}
        \hfill
    \begin{subfigure}[t]{0.41\textwidth} \hspace{-3em}
        \raisebox{33pt}{\includegraphics[width=\textwidth]{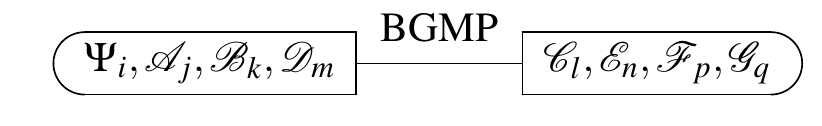}}
    \end{subfigure}
    \caption{Every closed circuit is equivalent to composing a preparation with an observation.}
    \label{f:EQU}
\end{figure}

\bigskip
A complete formulation of an OPT is given by providing the mathematical description of its systems and its transformations. Within the OPT framework, the von Neumann axiomatization of quantum theory\footnote{Here we are considering the quantum theory of abstract systems, without the original mechanical connotation of the theory, in which case the theory is usually referred to as {\em Quantum Mechanics}.} can be formulated as follows:\footnote{We use the following common notation: $\Bnd(\sH)$  ($\Bnd^+(\sH)$) denotes the set of bounded (positive) operators on the Hilbert space $\sH$, $T(\sH)$ ($\T^+(\sH)$) denotes the set of (positive) trace-class operators over $\sH$, $\T_{\leq1}^+(\sH)$ denotes the set of positive trace-class operators $\rho$ over $\sH$ with $\Tr[\rho]\leq 1$, $\Trn(\rA\to\rB)$ is the set of transformations from system $\rA$ to system $\rB$,  $\St(\rA)$ is the set of states of system $\rA$,  $\Eff(\rA)$ its set of effects. All maps in the following are linear.}
\begin{itemize}
\item[] {\bf Quantum Theory:} each {\em system} $\rA$ is associated to an Hilbert space $\sH_\rA$, and for the trivial system one has $\sH_\tI=\Cmplx$. A {\em transformation} $\tT\in\Trn(\rA\to\rB)$ is described by a completely positive (CP) trace-non-increasing map from $\T(\sH_\rA)$ to $\T(\sH_\rB)$, the transformation being deterministic when the map is trace-preserving. 
\end{itemize}
As mentioned, states and effects are special types of transformations, having trivial input system and trivial output system, respectively. We then have 
\begin{enumerate}
\item The set of states of system $\rA$ is $\St(\rA):=\Trn(\rI\to\rA)$. It follows that the states of $\rA$ are represented by positive maps\footnote{For trivial input or output system the CP map is simply a positive map.} from $[0,1]$ to $\T_{\leq1}^+(\sH_\rA)$,  the deterministic states corresponding to unit trace. In particular $\St(\rI)\equiv\T_{\leq1}^+(\Cmplx)\equiv[0,1]$ are probability values. 
\smallskip
\item The set of effects of system $\rA$ is $\Eff(\rA)=\Trn(\rA\to\rI)$. It follows that the effects of $\rA$ are positive maps from $\T^+_{\leq 1}(\sH_\rA)$ to $[0,1]$, hence they are functionals of the form $\varepsilon_i(\cdot)=\Tr_\rA[\cdot E]$ where $\Tr_\rA$ denotes the partial trace over $\sH_\rA$ and $\Bnd^+(\sH_\rA)\ni E\leq I_\rA$. $\Tr_\rA$ is the only deterministic effect of $\rA$.
\end{enumerate}
\def\dim{\operatorname{dim}}
This is all what we need to know about quantum theory. 

We need now to remark that classical theory can also be formulated as an OPT \cite{CUPDCP},  and it is a restriction of quantum theory to maximal simplexes of states contained in the quantum convex sets of states, and restricting transformations accordingly.\footnote{Classical theory reformulated as an OPT is the restriction of quantum theory, in which the set of states $\St(\rA)$ is limited to the simplex $\mathsf{S}^{\dim(\sH_A)}$ that is the convex hull of a fixed maximal set of orthogonal pure quantum states, and restricting also transformations accordingly. It is easy to see that the theory can be formulated as follows
\begin{itemize}
\item[] {\bf Classical Theory:} It is the OPT restriction of quantum theory corresponding to choosing $\Trn(\rA\to\rB)=\mathsf{T}(n_\rA, n_\rB)$ the set of $n_\rA\times  n_\rB$ sub-stochastic Markov matrices, with $n_\rA=\dim(\sH_\rA)$.
\end{itemize}
}

\subsection{The causality theorem}\label{s:causal2}
For an OPT the causality principle is stated as follows \cite{CUPDCP}:

\begin{itemize}
\item[]{\bf Causality principle:} The probability of  preparations is independent on the choice of observations.
\end{itemize}

We now show that the causality principle is a theorem of quantum theory.

\begin{theorem}\label{the:QTcausal}
Quantum theory satisfies the causality principle.
\end{theorem}
\begin{proof}
Since $\Tr_\rA$ is the unique deterministic effect for $\rA$, one has $\sum_{j\in\rY}\{E_j\}=I_\rA$ for every observation test. The joint probability of preparation and observation is given by 
\begin{equation}
p(i,j):=\Tr[E_j\rho_i],
\end{equation}
and the marginal probability of preparation is 
\begin{equation}\label{maracas}
\sum_{j\in\rY}p(i,j)=\sum_{j\in\rY}\Tr[\rho_i E_j]=\Tr[\rho_i],
\end{equation}
which is independent on the choice of observation $\{E_j\}_{j\in\rY}$.
\end{proof}
\medskip
Notice the asymmetry between preparations and observations. Indeed, differently from Eq. (\ref{maracas}), the marginal probability of  observation is
\begin{equation}
\sum_{i\in\rX}p(i,j)=\sum_{i\in\rX}\Tr[\rho_iE_j]=\Tr[\rho_\rX E_j]
\end{equation}
which generally depends on the preparation $\{\rho_i\}_{i\in X}$ through $\rho_\rX$.

One can easily realise that for a general OPT an equivalent form of the causality postulate is the following.

\begin{itemize}
\item[]{\bf Causality principle (equivalent form):} An OPT satisfies the causality postulate iff the deterministic effect is unique, and depends only on the system of the theory.
\end{itemize}
Indeed, as seen in the proof of Theorem \ref{the:QTcausal}, taking the marginal probability of preparation corresponds to coarse-grain the observation test into a deterministic effect. Therefore, independence of deterministic effect on the observation test is equivalent to independence of the marginal probability on the choice of the observation test. 

In the book \cite{CUPDCP} it is also shown how the causality principle is also equivalent to the possibility of normalising states, which in turn is equivalent to the possibility of preparing deterministically any state via post-selection.\footnote{As remarked in \cite{d2013determinism}, if the theory is not causal there exist states which cannot be prepared with certainty. In other words, the post-selection procedure is not available.} Finally, it is shown how the causality postulate restricts the convex structure of the theory, e.g. with the convex set of states shaped as an hyperplane-truncated cones and those of effects shaped as spindlers.

\subsection{Causality is a partial ordering between events}
The following definition establishes a partial ordering among events in a circuit.
\begin{definition}[Partial ordering between events]\label{def:porder}
For two events $\tA$ and $\tB$ in the same circuit, we say that  {\em event $\tA$ precedes event $\tB$} if there is a undirected input-output path of the circuit that connects an output system of  $\tA$ with an input system of $\tB$. We equivalently say that that {\em $\tB$ follows $\tA$}. 
\end{definition}
 Definition \ref{def:porder} holds equally well by substituting the word ``event'' with ``test''. It introduces the notions of {\em input} and {\em output cones} of a given event (test). 
 \begin{definition}[Input and output cones of an event] The {\em input cone of event $\tA$} is the set of events preceding $\tA$.
 Analogously the {\em output cone of event $\tA$} is the set of events following $\tA$. 
 \end{definition}
 In the presence of the Causality assumption in a general OPT, the input and output causal cones are the respective {\em causal cones} (figure\ref{f:ccones}) .
  Moreover, if we associate the input-output direction with the {\em arrow of time}, the two cones become the equivalent of the {\em past cone} and {\em future cone}, respectively. 
\begin{figure}[h]
     \centering
    \begin{subfigure}[t]{0.58\textwidth}
        \raisebox{0pt}{\includegraphics[width=\textwidth]{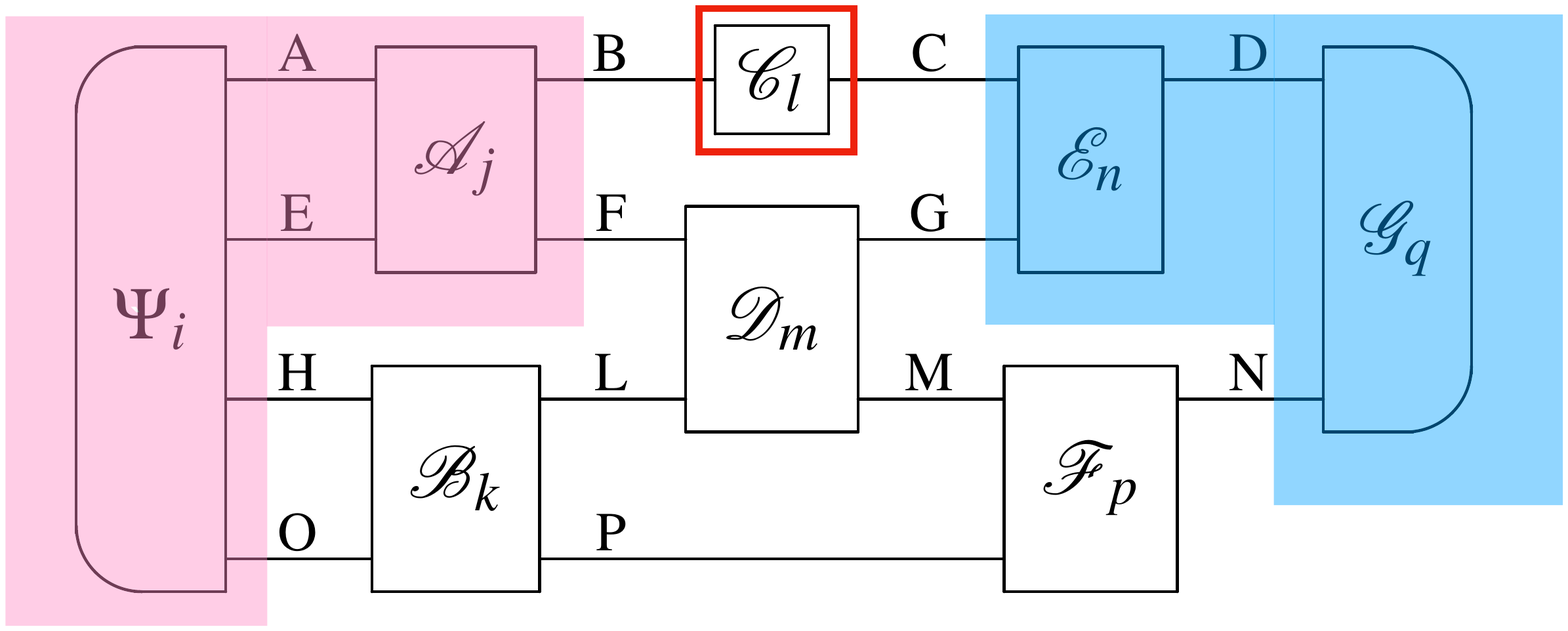}}
\caption{Past (pink) and future (blue) causal cones of event $\tC_l$ in the circuit of Fig. \ref{f:EQU}.}
    \end{subfigure}
    \hfill
    \begin{subfigure}[t]{0.36\textwidth}\hspace{1em}
        \raisebox{20pt}{\includegraphics[width=\textwidth]{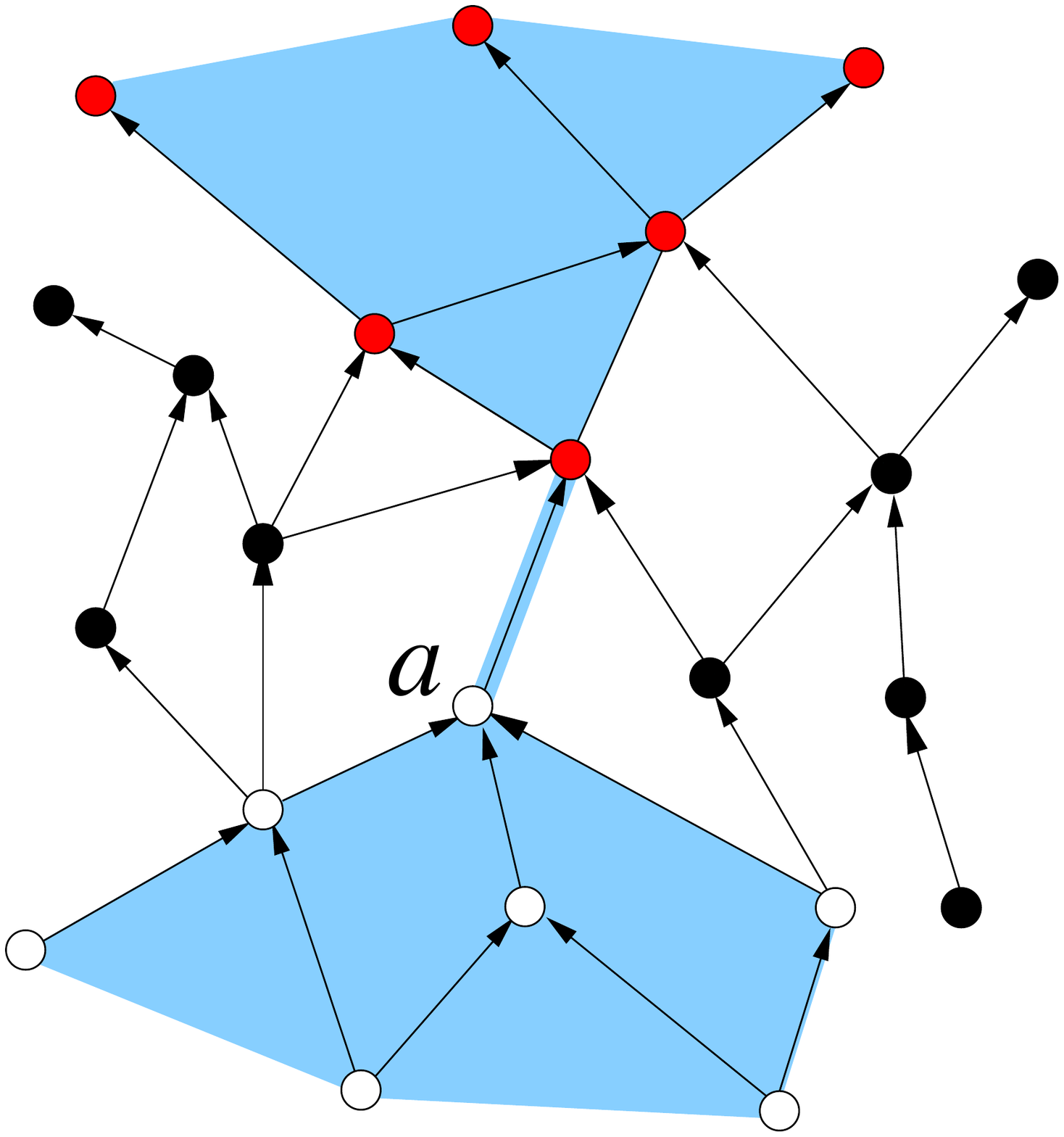}}
       \caption{Past and future causal cone the event $a$ in customary causal-network representation.}
    \end{subfigure}
    \caption{Causal cones of an event.}
    \label{f:ccones}
\end{figure}
 
 \subsection{Causality in Special Relativity}
With the above notions in mind, it is now immediate to see that the causality principle is equivalent to the following statement
\begin{corollary}[No-signaling from the future]\label{c:caus} An OPT is causal iff the marginal probability of any test is independent on the choice of any test that does not belong to its past cone.
\end{corollary}
\begin{proof}
The statement is equivalent to say that the OPT is causal, if  and only if for any test $\tA$ that does not follow a test $\tB$, one has that the marginal probability distribution of test $\tA$ is independent of the choice of test $\tB$. This statement implies the causality principle, since the latter is a special case. The reverse implication is also true, since as we have seen, any closed circuit can be always split into a preparation and an observation test, with the preparation containing test $\tA$ and the observation containing test $\tB$.
\end{proof}
An tautological alternative statement of Corollary \ref{c:caus} is the following
\begin{corollary}[No-signaling from the future]\label{c:caus2} An OPT is causal iff the marginal probability of any event is independent on which events can occur outside its past cone.
\end{corollary}
In the formulation of  Corollary \ref{c:caus2} causality has exactly the form used in special relativity, within a deterministic scenario in a classical case. Indeed causality is the existence of a partial ordering between events, which defines a ``preceding'' (and ``following'') cone, and, viceversa, the cone can be used to define the partial ordering. Thus logically the causality notion of OPTs (and hence of quantum theory) is the same notion that we use in special relativity, although without the need of introducing the Minkowski structure. 

\subsection{Causality is falsifiable science}
We see now that causality, being a theorem of quantum theory, is also experimentally falsifiable, namely it allows us to formulate predictions that are falsifiable. Therefore, according to the Popper's demarcation criterion\cite{popp-logi}, causality is in all respects scientific theory. 

An experiment for falsifying causality can be designed simply by building up a cascade of two quantum measurements, e.g. two Stern-Gerlach apparata with pinhole-detection of particles, with the possibility of changing the second apparatus at the output, .e.g. changing the orientation of the magnetic field gradient. Two apparata in cascade correspond to the sequential composition $\{\tB_j\tA_i\}_{(i,j)\in\rX\times\rY}\subset\Trn(\rA\to\rC)$ of the two tests $\hat\tA=\{\tA_i\}_{i\in\rX}\subset\Trn(\rA\to\rB)$ and $\hat\tB:=\{\tB_j\}_{j\in\rY}\subset\Trn(\rB\to\rC)$. Here test $\hat\tA$ belongs to the past cone of test $\hat\tB$, and viceversa test $\hat\tB$ belongs to the future cone of test $\hat\tA$. Using the rules of quantum theory we have the causality relations
\begin{align*}
&p_\tA(i):=\sum_{j\in\rY}p(i,j)=\sum_{j\in\rY}\Tr[\tB_j\tA_i(\rho)]=\Tr[\tB_\rY\tA_i(\rho)]=\Tr[\tA_i(\rho)],
\text{ independent on test }\{\tB_j\}_{j\in\rY}&\\
&p_\tB(j):=\sum_{i\in\rX}p(i,j)=\sum_{i\in\rX}\Tr[\tB_j\tA_i(\rho)]=\Tr[\tB_j\tA_\rX(\rho)],
\text{ generally dependent on test }\{\tA_i\}_{i\in\rX},&
\end{align*}
where we used the fact that the coarse graining $\tB_\rY:=\sum_{j\in\rY}\tB_j$ is trace preserving (true also for $\tA_\rX$). One can now consider $\rX=\{0,1\}$ with 
$\Tr[\tA_0(\rho)]=0$. Then, for example, the causality principle would be falsified if upon changing test $\hat\tB$ to a different test  $\hat\tB'$, one would have occurrence of the outcome $0$ for test $\hat\tA$.
\subsection{From Cinderella to Principle of Physics}
The causality principle that here we have derived as a theorem of quantum theory, in Refs. \cite{purification,QUIT-Arxiv} has been used instead as the first of six information-theoretical postulates to derive quantum theory in finite dimensions (a didactical derivation can be found in the textbook \cite{CUPDCP}). Recognising causality as a postulate for the theory has allowed for the first time to consider variations of the theory that do not satisfy it. This has spawn a full research line about a non-causal variations of quantum theory, ranging from the theory of {\em quantum combs} \cite{PhysRevA.88.022318} (which still has an interpretation in terms of an underlying causal theory) to theories with dynamical causal ordering 
\cite{Oreshkov:2012jc,Brukner:2014if,Brukner:2014vo,Brukner:2014vo,Baumeler:2015te}, and experimental tests have been also devised \cite{2015NatCo...6E7913P,Rambo}.

\section{Analysis of the causality-denial sindrome in classical physics}
We have seen the statement of the causality principle in terms of the independence of the marginal probability of preparation-tests on the choice of observation-tests, or, shortly: independence of preparations on observations. We argued that this statement distills all the guises in which causality appears in physics. We have seen that the principle is a theorem of quantum theory, but also classical theory must satisfy the principle, since it is a restriction of quantum theory. Then, why it seems that there is no causality principle in classical physics, as Carrol asserts? 

The answer is that within classical theory the principle of causality is trivialised  in many ways by the ``realism of the theory'', i.e. the one-to-one correspondence between elements of reality and elements of the theory. On one side, restricting to a deterministic scenario put causes and effects (preparations and observations) in one-to-one correspondence, thus trivialising their distinction. On the other side, classical theory makes the notion of  ``measurement'' irrelevant, since it is identified with the reading of the state, hence identifying ``state'' and ``measurement outcome''. Finally, in classical physics there is only one possible observation, whereas in quantum theory inequivalent observations are introduced by complementarity, and this breaks the classical identification between measurement and state. 

What then causality means in the classical case? Although in a classical world there is no choice of observation, yet one can chose the preparation,\footnote{The hypothesis of {\em superdeterminism} is discarded as methodologically non scientific, since the same falsifiability of a physical law requires the independent variation of physical parameters in experiments.}
 and this makes the formulation of Corollary \ref{c:caus2} appropriate. And, indeed, we have seen that this statement is shared by special relativity, namely it is the usual Einsteinian notion. With the same argument Corollary \ref{c:caus2} requires discarding advanced potentials in electromagnetism. 

Finally, it is worth recalling that causality is a genuinely {\em theoretical} notion, since it relies on a theory of connections among events, namely the connections are established by an underlying theoretical substratum. Thus, the Minkowski space-time structure  {\em per se} has no bearing on causal connections, but it is  an underlying theory, e.g. Maxwell theory, which establishes which events can be connected and which cannot. 

\section{Conclusions}
Causality establishes an arrow of time, going from preparation to observation. 
This corresponds to the usual arrow of time, the one that we use in our everyday experience. By denying causality, Sean Carrol denies such arrow of time, against the everyday experience. He writes\cite{bigpict}:
\\\par\noindent {\em The reason why there's a noticeable distinction between up and down for us isn't because of the nature of space; it's because we live in the vicinity of an extremely influential object: the Earth. Time works the same way. In our everyday world, time's arrow is unmistakable, and you would be forgiven for thinking that there is an intrinsic difference between past and future.} 
\\\par In order to reconcile experience, he then invokes the Past Hypothesis of David Albert 
\\\par\noindent {\em In reality, both directions of time are created equal. The reason why there's a noticeable distinction between past and future isn't because of the nature of time; it's because we live in the aftermath of an extremely influential event: the Big Bang.
 \ldots
The thing we need to add is an assumption about the initial condition of the observable universe, namely, that it was in a very low-entropy state. Philosopher David Albert has dubbed this assumption the Past Hypothesis.
\ldots
What we know is that this initially low entropy is responsible for the  ``thermodynamic''  arrow of time,\ldots}
\\\par The idea that there is no distinction between past, present, and future was of Einstein himself, who wrote in a letter to the family of Michele Besso, few days after Einstein learned of his death:
\\\par\noindent {\em Now he has departed from this strange world a little ahead of me. That means nothing. People like us, who believe in physics, know that the distinction between past, present and future is only a stubbornly persistent illusion.}\cite{Einstein-isaacson}
\\\par Einstein here denies the existence of a time-arrow. However, one should not forget that he has always been a stubborn opponent of quantum theory in his battle with Bohr. But  everybody nowadays would agree that who lost the battle was Einstein.
\dataccess{For a recent review on the causality postulate, see Chapt. 5 of the book \cite{CUPDCP}.}
\competing{I declare I have no competing interests.}
\funding{This work was made possible through the support of a grant from the John Templeton Foundation (Grant ID\# 43796: {\em A Quantum-Digital Universe}, and Grant ID\# 60609: {\em Quantum Causal Structures}). The opinions expressed in this publication are those of the author and do not necessarily reflect the views of the John Templeton Foundation. The manuscript has been conceived and drafted at the Kavli Institute for Theoretical Physics, University of California, Santa Barbara, and supported in part by the National Science Foundation under Grant No. NSF PHY11-25915.} 
\ack{The main purpose of the present paper is that of making the scientific community aware of the existence of a falsifiable causality principle holding for the whole  contemporary physics. The paper is based on previous results, mostly in collaboration with P. Perinotti and G. Chiribella, published in the book \cite{CUPDCP} and in the original papers 
\cite{purification,QUIT-Arxiv,d2013determinism}. In the present paper causality is derived as a theorem of quantum theory, whereas in the mentioned works it is a postulate for the theory. A preliminary version of the causality postulate has been given in Ref.~\cite{PhilosophyQuantumInformationEntanglement2} rephrased as 
``no-signaling from the future'', of which the first embryo can be find in a footnote in Ref.~\cite{asd}.}

\bibliographystyle{iopart-num}
\bibliography{Causality-Phyl-TransA}
\end{document}